\documentclass{ifacconf}
\usepackage{subcaption}
\usepackage{amsfonts}
\usepackage{amssymb}
\usepackage{amsmath}
\makeatletter
\let\theoremstyle\@undefined
\makeatother
\usepackage{amsthm}
\usepackage{xcolor}
\usepackage{graphicx}      
\usepackage{natbib}        
\usepackage{cuted}
\usepackage{siunitx}
\usepackage{tikz}
\usetikzlibrary{arrows.meta, positioning, shapes, calc}
\usepackage{tabularx}
\newtheorem{theorem}{Theorem}[section]

\newtheorem{Definition}[theorem]{Definition}
\newtheorem{proposition}[theorem]{Proposition}
\newtheorem{lemma}{Lemma}
\newtheorem{Assumption}{Assumption}

\newcommand{\dom}{\ensuremath{\mathrm{dom}}}
\newcommand{\rline}{\ensuremath{\mathbb{R}}}
\newcommand{\cA}{\ensuremath{\mathcal{A}}}
\newcommand{\cB}{\ensuremath{\mathcal{B}}}
\newcommand{\cC}{\ensuremath{\mathcal{C}}}

\newcommand{\cE}{\ensuremath{\mathcal{E}}}
\newcommand{\cF}{\ensuremath{\mathcal{F}}}
\newcommand{\cG}{\ensuremath{\mathcal{G}}}
\newcommand{\cH}{\ensuremath{\mathcal{H}}}

\newcommand{\cS}{\ensuremath{\mathcal{S}}}

\newcommand{\bR}{\ensuremath{\mathbf{R}}}

\renewcommand{\bf}{\ensuremath{\mathbf{f}}}

\newcommand{\bu}{\ensuremath{\mathbf{u}}}

\newcommand{\bx}{\ensuremath{\mathbf{x}}}
\newcommand{\by}{\ensuremath{\mathbf{y}}}
\newcommand{\bz}{\ensuremath{\mathbf{z}}}

\begin{document}
\begin{frontmatter}

\title{Irreversible Port-Hamiltonian Formulations for 1-Dimensional fluid systems \thanksref{footnoteinfo}} 

\thanks[footnoteinfo]{This work has been achieved in the frame of the EIPHI Graduate school (contract ``ANR-17-EURE-0002") and ANID funded projects
CIA250006 and FONDECYT 1231896}
\author[First]{Ahlam Ouardi} 
\author[Second]{Arijit Sarkar} 
\author[Third]{Hector Ramirez}
\author[Fourth]{Yann Le Gorrec}

\address[First]{EMINES- School of Industrial Management, University Mohammed VI Polytechnic, Benguerir, Morocco}
\address[Second]{Department of Control Systems and Network Control Technology, Brandenburg University of Technology Cottbus-Senftenberg, Germany}
\address[Third]{Departamento de Electronica, Universidad Tecnica Federico Santa Maria, Valparaiso, 2390123, Chile}
\address[Fourth]{Université Marie et Louis Pasteur, SUPMICROTECH, CNRS, institut FEMTO-ST, F-25000 Besançon, France}

\begin{abstract}                
The Irreversible Port-Hamiltonian Systems (IPHS) framework is extended to the modelling of non-isentropic fluids with viscous dissipation in the Eulerian description. Building on earlier IPHS formulations for diffusion-driven and non-convective distributed systems, it is shown that convective transport can be consistently encompassed by the framework by modifying the underlying differential operators. After revisiting the constitutive relations of non-isentropic fluids in both Eulerian and Lagrangian coordinates, it is demonstrate how these systems fit within an extended IPHS formulation. Furthermore, an extended parametrisation of the boundary port variables which ensures that the first and second laws of Thermodynamics are fulfilled allows to define a general class of boundary controlled IPHS. 
\end{abstract}


\end{frontmatter}
\section{Introduction}

Irreversible port-Hamiltonian Systems (IPHS) formulations were first introduced in \citep{Ramirez_CES_2013,Ramirez_EJC_2013} and later developed as an alternative to pseudo port Hamiltonian formulations \citep{Favache_CES10,Hoang_JPC_2012} or metriplectic-GENERIC \citep{GENERIC_I_PhysRevE} formulations, to extend classical port-Hamiltonian theory to the modelling and control of irreversible thermodynamic systems whose dynamics depend strongly on temperature. Such systems include many cutting-edge engineering applications involving active materials, fluid flows, or energy conversion processes (e.g., chemical or nuclear reactions) \citep{MoraPoF2021,CardosoJCF2025}.
In contrast to classical port-Hamiltonian approaches, IPHS use the total energy as the potential function and incorporate the entropy balance directly into the system dynamics. This leads to a nonlinear structure that depends on the co-state variables and no longer satisfies the Jacobi identities. Nevertheless, the modulated skew-symmetric interconnection matrix still reflects energy conservation in accordance with the first law of thermodynamics, while the entropy balance enforces the second law. A nice property is that the modulation function directly depends on the topology and the irreversible thermodynamic forces of the irreversible processes involved in the systems. 
These features make IPHS a very promising framework, particularly for control design \citep{Ramirez_Automatica_2016}, which can be reinterpreted in terms of energy shaping and entropy assignment.
IPHS formulations have recently been extended to distributed-parameter systems defined over one-dimensional spatial domains \citep{Ramirez2022CES,Ramirez2022Entropy}. This extension has proven particularly effective for representing diffusion-type systems, such as the heat equation, and for deriving boundary control laws that ensure their stabilization \citep{Mora2024SCL}. The same parameterization can be applied directly to non-convective hyperbolic systems, such as deformable mechanical structures, but it is not suitable for convective systems. In the latter case, the assumption of local thermodynamic equilibrium requires tracking the material particles, which naturally leads to formulations in the Lagrangian framework.
The aim of this paper is to extend the IPHS formulations to convective systems formulated in the Eulerian description, focusing on the simple yet illustrative case of a non isentropic fluid with viscous damping.

The paper is organized as follows: Section \ref{Sec:nonisentropic} provides the governing equations of 1D non isentropic compressible fluids in both Eulerian and Lagrangian coordinates. In Section \ref{Sec:IPHSfluids} are derived the irreversible port Hamiltonian formulations in these two different frames, leading to the general IPHS formulation in Section \ref{Sec:IPHSgeneral}. Some conclusions and perspectives are given in Section \ref{Sec:conclusions}.

\section{Non-isentropic compressible fluids}
\label{Sec:nonisentropic}
In this section, we focus on non-isentropic compressible fluids within a one-dimensional spatial domain. The behavior of fluid flows can be described from two different reference frames. In the Eulerian frame, or fixed coordinate system, the focus is on observing the fluid as it moves past a stationary point in space. In contrast, the Lagrangian frame, or moving coordinate system, follows the motion of individual fluid particles. Fig. \ref{fig:Lagrange_Euler} vividly depicts the differences between the two coordinate frames.
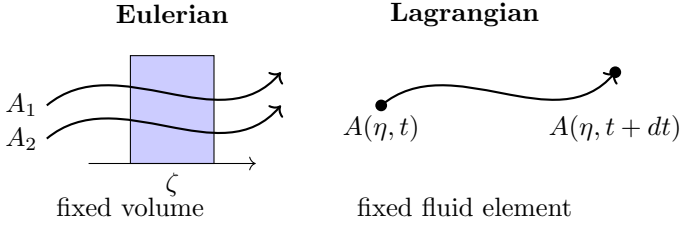
\begin{figure}
\begin{center}
\begin{tikzpicture}[scale=1.1]



\node at (5,2.5) {\textbf{Lagrangian}};
\fill (4,1.4)  circle (2pt);
\fill (6.8,1.8) circle (2pt);
\node[below] at (4,1.4) {$A(\eta,t)$};
\node[below] at (6.8,1.4) {$A(\eta,t+dt)$};
\draw[->, thick] (4,1.4)  to[out=40,in=230]  (6.8,1.8);
\node[below] at (5,0.4) {fixed fluid element};

\node at (1.5,2.5) {\textbf{Eulerian}};
\draw[fill=blue!20] (1,0.7) rectangle (2,2);
\node[below] at (1,0.4) {fixed volume};
\node[left] at (0,1.4) {$A_1$};
\node[left] at (0,1) {$A_2$};
\draw[->, thick] (0,1.4)  to[out=40,in=230]  (2.8,1.8);
\draw[->, thick] (0,1)  to[out=40,in=230]  (2.8,1.4);
\draw[->] (0.5,0.7) -- (2.5,0.7);
\node[below] at (1.5,0.7) {$\zeta$};

\end{tikzpicture}
\end{center}
\caption{Eulerian and Lagragian coordinates.\label{fig:Lagrange_Euler}}
\end{figure}
Now, we derive the governing equations of fluid flow using mass, momentum and energy (entropy) balance in both the reference frames.
\subsection{Governing equations : Eulerian coordinate}
We denote by $\rho = \rho(\zeta,t)$, $v = v(\zeta,t)$, $p = p(\zeta,t)$ the fluid mass density per unit length, velocity and pressure, respectively, with $\zeta \in \left[a,b \right]$. By applying Reynold's transport theorem and divergence theorem, the conservation of mass can be represented as the continuity equation\footnote{For simplicity, we use $\partial_k$ to represent $\frac{\partial}{\partial k}$ or $D_t$ to denote $\frac{D}{D t}$}
\begin{align}\label{eq:mass_conservation_Euler}
    \frac{\partial \rho }{ \partial t} +  \partial_\zeta\left(\rho v\right) = 0.
\end{align}
Next, using Newton's second law for the material volume neglecting the gravitational forces and then applying Reynold's transport theorem and divergence theorem, we obtain the Cauchy momentum equation 
\begin{align}\label{eq:momentum_Cauchy}
    \partial_t (\rho v) + \partial_\zeta(\rho v^2) =  \partial_\zeta \sigma,
\end{align}
where  $\sigma$ represents the stress tensor. Consider $\sigma = -p - \tau$, where $p$ is the thermodynamic pressure, $\tau$ is the viscous stress. Now, considering a Newtonian fluid, we have the viscous stress $\tau = -\mu \partial_\zeta v$, where $\mu$ is the coefficient of dynamic viscosity. Using the continuity equation \eqref{eq:mass_conservation_Euler} we have the following momentum balance equation 
\begin{align}\label{eq:momentum_balance}
    \rho \partial_t v = -\rho v \partial_\zeta v - \partial_\zeta p + \partial_\zeta \left( \mu \, \partial_\zeta v\right).
\end{align}
Now, using Gibbs' equation, i.e. $du=-pd\Phi+Tds$ where $\Phi=\frac{1}{\rho}$ is the specific volume, we have
\begin{align}\label{eq:Gibbs'}
    \partial_\zeta u = -p\partial_\zeta \left(\frac{1}{\rho}\right) + T \partial_\zeta s.
\end{align}
By using the specific enthalpy $h := u + \frac{p}{\rho}$ we have
\begin{align}
    \frac{1}{\rho}\partial_\zeta p = \partial_\zeta h - T\partial_\zeta s.
\end{align}
From \eqref{eq:momentum_balance}, we have
\begin{align}\label{eq:velocity_equation}
    \partial_t v = -v \, \partial_\zeta v - \partial_\zeta h + T \, \partial_\zeta s + \frac{1}{\rho}\partial_\zeta(\mu  \partial_\zeta v).
\end{align}
Finally, the rate of change of the internal energy per unit mass $u$ is equal to the sum of the heat conduction due to the heat flux $q$ and the mechanical work, and is given by
\begin{align}
     \partial_t u &=-v \partial_\zeta u -\frac{1}{\rho}\partial_\zeta q - \frac{p}{\rho} \partial_\zeta v - \frac{\tau}{\rho} \partial_\zeta v,
\end{align}
 The heat flux follows Fourier's law, given by $q := -k\partial_\zeta T$ where $k$ is the thermal conductivity. Using $\sigma = -p + \tau$, we have
\begin{align}\label{eq:energy_balance_Euler}
    \partial_t u = -v\partial_\zeta u +  \frac{1}{\rho}\partial_\zeta(k\partial_\zeta T) - \frac{p}{\rho} \partial_\zeta v + \frac{\mu}{\rho}(\partial_\zeta v)^2.
\end{align}
Hence, \eqref{eq:mass_conservation_Euler}, \eqref{eq:velocity_equation}, and \eqref{eq:energy_balance_Euler} constitute the governing differential equations of a compressible non-isentropic fluid in Eulerian coordinates flowing in one dimensional spatial domain. Again, using Gibbs' equation we have\footnote{From local thermodynamic equilibrium assumption the Gibbs' equation has to be written following the matter, i.e. using material derivative. }
\begin{align}
    D_t s = \frac{1}{T} D_t u + \frac{p}{T} D_t \left(\frac{1}{\rho}\right).
\end{align}
where $D_t := \partial_t + v\partial_\zeta$ is the material derivative.
Then, expanding the material derivative, we have the following governing equation in terms of specific entropy $s$ instead of specific internal energy $u$
\begin{align}\label{eq:entropy_Euler}
    \partial_t s = -v \, \partial_\zeta s 
+ \frac{1}{\rho T} \partial_\zeta (k \partial_\zeta T)
+ \frac{\mu}{\rho T} (\partial_\zeta v)^2.
\end{align}
To summarize, the governing equations of a non isentropic fluid subject to viscous damping and heat conduction in Eulerian coordinates are:

\begin{align}
\partial_t \rho &=-  \partial_\zeta(\rho v) \\
\partial_t v &= -v \, \partial_\zeta v - \partial_\zeta h + T \, \partial_\zeta s + \frac{1}{\rho}\partial_\zeta(\mu  \partial_\zeta v)\\
    \partial_t s &= -v \, \partial_\zeta s 
+ \frac{1}{\rho T} \partial_\zeta (k \partial_\zeta T)
+ \frac{\mu}{\rho T} (\partial_\zeta v)^2.
\end{align}
and the total energy is given by
\[
{\mathcal H}(\rho,v,s)= \int_a^b \left( \frac{1}{2}\rho v^2 + \rho u(\rho,s) \right) d\zeta
\]
\subsection{Governing equations: Lagrangian coordinate}
Consider now a small fluid element as in Figure \ref{Fig:element}.
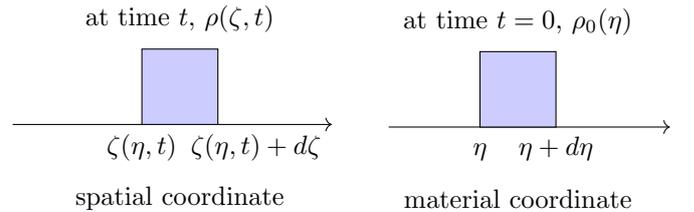
\begin{figure}[h]
\begin{center}
\begin{tikzpicture}[scale=1, baseline=(current bounding box.center)]
\draw[->] (-1.2,0) -- (3,0);

\fill[blue!20] (0.5,0) -- (0.5,1) -- (1.5,1) -- (1.5,0) -- cycle;
\draw (0.5,0) -- (0.5,1) -- (1.5,1) -- (1.5,0) -- cycle;

\node at (1,1.4) {at time $t$, $\rho(\zeta,t)$};
\node[below] at (0.5,0) {$\zeta(\eta,t)$};
\node[below] at (2,0) {$\zeta(\eta,t)+d\zeta$};

\node[below] at (1,-0.7) {spatial coordinate};

\end{tikzpicture}
\hspace{0.5cm}
\begin{tikzpicture}[scale=1, baseline=(current bounding box.center)]
\draw[->] (-1.2,0) -- (2.5,0);
\fill[blue!20] (0,0) -- (0,1) -- (1,1) -- (1,0) -- cycle;
\draw (0,0) -- (0,1) -- (1,1) -- (1,0) -- cycle;
\node at (0.5,1.4) {at time $t=0$, $\rho_0(\eta)$};
\node[below] at (0,-0.1) {$\eta$};
\node[below] at (1,0) {$\eta+d\eta$};

\node[below] at (0.5,-0.7) {material coordinate};

\end{tikzpicture}
\end{center}
\caption{Small fluid element in spatial and material coordinates.\label{Fig:element}}
\end{figure}

Mass conservation principle \citep{malvern1969introduction} implies
\begin{equation}
dm_0 = dm_t 
\end{equation}
leading to
\begin{equation}
\rho_0(\eta)\, d\eta = \rho(\zeta,t)\, d\zeta.
\end{equation}

Thus,
\begin{equation}
 \frac{\partial \eta}{\partial \zeta} = \frac{\rho(\zeta,t)}{\rho_0(\eta)}.
\end{equation}

Let's calculate the mass equation in Lagrangian coordinates. From \eqref{eq:mass_conservation_Euler}, we have the mass balance equation in Eulerian coordinate as 
\[\partial_{t}\rho = -\partial_{\zeta}(\rho v)\]

Using the material derivative we have 
\begin{align}
    D_{t}\rho &= \partial_{t}\rho + v\partial_{\zeta}\rho \nonumber \\ &= -\partial_{\zeta}(\rho v) + v\partial_{\zeta}\rho \nonumber \\ &= -\rho \partial_{\zeta}v \nonumber \\
     &= -\rho (\partial_\zeta \eta )(\partial_\eta v) \nonumber\\ & = -\frac{\rho^2}{\rho_0}\partial_\eta v \nonumber
     \end{align}
     leading to
    \begin{equation} D_{t}\left(\frac{1}{\rho}\right) = \frac{1}{\rho_{0}}\partial_{\eta}v \label{eq:density_Lagrange}.
\end{equation}
By performing similar calculations as outlined above and applying the definition of the material derivative, we can derive the velocity and entropy equations in Lagrangian coordinates as follows:
\begin{align}
    D_t v &= -\frac{1}{\rho_0}\partial_\eta p + \frac{\mu}{\rho_0^2}\partial_\eta^2 v \label{eq:velocity_Lagrange}\\
    D_t s & = \frac{1}{\rho_0 T}\left( k\, \partial_\eta^2 T + \mu (\partial_\eta v)^2 \right) \label{eq:entropy_Lagrange}.
\end{align}
Then, \eqref{eq:density_Lagrange}, \eqref{eq:velocity_Lagrange}, and \eqref{eq:entropy_Lagrange} constitute the governing equations of a compressible non-isentropic fluid in Lagrangian coordinates flowing in one dimensional spatial domain via replacing the material derivative $D_t$ with $\partial_t$.
\section{IPHS formulations of non isentropic compressible fluids}
\label{Sec:IPHSfluids}In this section we derive the IPHS formulation of non isentropic fluids in both Eulerian and Lagrangian coordinates.
\subsection{IPHS fromulation in Eulerian coordinates}


In IPHS formulations \citep{Ramirez2022CES}, the transport and dissipative mechanisms are represented in a way that explicitly highlights energy flows and entropy production. This explicit representation is crucial for both physical interpretation and control design. With this goal in mind, we choose as state variables the energy variables  
$\bx := 
\begin{pmatrix}
x^\top & s
\end{pmatrix}^\top$, where $x := (\rho \; v)^\top$.
The total energy (Hamiltonian) of the fluid is defined as 
\begin{align}
    \label{ham}
\cH(\bx) &= \int_{a}^b \left( \frac{1}{2} \rho v^2 + \rho u(\rho,s) \right) \, d\zeta\\ &=\int_{a}^b \left( \frac{\rho v^2}{2} + \rho h - p \right) \, d\zeta.
\end{align}

The co-energy (effort) variables are defined as 
$$\begin{aligned}
\delta_{\rho} \cH &= \frac{1}{2} v^2 + h \\[4pt]
\delta_{v} \cH &= \rho v \\[4pt]
\delta_{s} \cH &= \frac{\delta \cH}{\delta u}.\frac{\partial u}{\partial s} = \rho T \quad \quad \quad\text{(since $du = -p \, d\left(\frac{1}{\rho}\right) + T \, ds$)}
\end{aligned}$$

Now, we can rewrite the governing equations \eqref{eq:mass_conservation_Euler}, \eqref{eq:velocity_equation}, and \eqref{eq:entropy_Euler} in terms of energy and co-energy variables as follows
\begin{align*}
\partial_{t} \rho &= - \partial_{\zeta} \delta_{v} \cH \\
\partial_{t} v &= - \partial_{\zeta} \delta_{\rho } \cH + \frac{1}{\rho} \partial_{\zeta} s\delta_{s} \cH  + \frac{\mu }{\rho} \partial_{\zeta} \left( \frac{\partial_{\zeta} v \delta_{s} \cH}{\rho T} \right) \\
\partial_{t} s &= - \frac{1}{\rho} \partial_{\zeta} s \delta_{v} \cH  + \frac{\mu }{\rho T} \partial_{\zeta} v \partial_{\zeta} \left( \frac{\delta_{v} \cH}{\rho } \right) + \frac{k}{\rho T} \partial_{\zeta} \left(  \partial_{\zeta} \left( \frac{1}{\rho} \delta_{s} \cH \right) \right)
\end{align*}
or equivalently as
\begin{equation}
\begin{pmatrix}
\partial_{t} \rho \\
\partial_{t} v \\
\partial_{t} s
\end{pmatrix} = 
\begin{pmatrix}
0 & -\partial_{\zeta} & 0 \\
-\partial_{\zeta} & 0 & \frac{1}{\rho}\partial_\zeta s+\cA \\
0 & -\frac{1}{\rho}\partial_\zeta s-\cB & \cC
\end{pmatrix}
\begin{pmatrix}
\delta_{\rho} \cH \\
\delta_{v} \cH \\
\delta_{s} \cH
\end{pmatrix} 
\end{equation}

Where the operators $\cA$ and $\cB$ are defined from Hilbert space $H^1({[a,b]})$ to $H^1({[a,b]})$ as follows
\begin{align}
\cA &: f \mapsto \frac{\mu }{\rho} \partial_{\zeta} \left( \frac{\partial_{\zeta} v}{\rho T} f \right)\\
\cB &: g  \mapsto -\frac{\mu }{\rho} \frac{\partial_{\zeta} v}{T} \partial_{\zeta} \left( \frac{g}{\rho} \right),
\end{align}
and the operator $\cC$ is given by 
\begin{align}
    \cC\delta_sH=\frac{k}{\rho T} \partial_{\zeta} \left( \delta_{\zeta} \left(\frac{1}{\rho} \delta_s \cH \right) \right).
\end{align}
For $f \in \dom(\cC)$, 
\begin{align}
\cC :f \mapsto \frac{1}{\rho T} \partial_{\zeta} \left(k \partial_{\zeta} \frac{f}{\rho}\right)&=\frac{1}{\rho } \partial_{\zeta} \left(\frac{k}{T} \partial_{\zeta} \frac{f}{\rho}\right)-\\
&\frac{1}{\rho } \partial_{\zeta}\left(\frac{1}{T}\right) \left(k \partial_{\zeta} \frac{f}{\rho}\right).
\end{align}
\begin{lemma}
    $\cA$ is the formal adjoint of $\cA$, i.e. $\cB=\cA^\ast$.
\end{lemma}
\begin{proof} We have
    \begin{align}
    \langle \cA f,g\rangle_{L^{2}({[a,b]})} &= \int_a^b  \frac{\mu }{\rho} \partial_{\zeta} \left( f\frac{\partial_{\zeta} v}{T \rho} \right) g  d\zeta \nonumber\\
&= - \int_a^b f \mu\frac{\partial_{\zeta} v}{\rho T} \partial_{\zeta} \left(  \frac{g}{\rho}\right)d\zeta +\left[ \frac{\mu}{\rho} g \frac{\partial_{\zeta} v}{T} \frac{f}{\rho} \right]_{a}^b \nonumber\\
&=\langle f,\cB g \rangle_{L^{2}({[a,b]})}+\left[ \frac{\mu}{\rho} g \frac{\partial_{\zeta} v}{T} \frac{f}{\rho} \right]_{a}^b.\nonumber \\
&=\langle f,\cA^\ast g \rangle_{L^{2}({[a,b]})}+\left[ \frac{\mu}{\rho} g \frac{\partial_{\zeta} v}{T} \frac{f}{\rho} \right]_{a}^b.\nonumber
\end{align}
If boundary conditions are considered equal to $0$, then $\langle \cA f,g\rangle_{L^{2}({[a,b]})} = \langle f,\cB g\rangle_{L^{2}({[a,b]})}$, i.e., $\cB$ is the formal adjoint of $\cA$, i.e. $\cB=\cA^\ast$.
\end{proof}

\begin{lemma}
    $\cC$ is a formally skew-adjoint operator.
\end{lemma}
\begin{proof}
The proof is a direct application of \citep[Theorem A.6]{luis}.
\end{proof}
We then have 
\begin{equation}
\label{IPHSEulerianfin}
\begin{pmatrix}
\partial_{t} \rho \\
\partial_{t} v \\
\partial_{t} s
\end{pmatrix} = 
\begin{pmatrix}
0 & -\partial_{\zeta} & 0 \\
-\partial_{\zeta} & 0 & \frac{1}{\rho}\partial_\zeta s+\cA \\
0 & -\frac{1}{\rho}\partial_\zeta s-\cA^\ast & \cC
\end{pmatrix}
\begin{pmatrix}
\delta_{\rho} \cH \\
\delta_{v} \cH \\
\delta_{s} \cH
\end{pmatrix} 
\end{equation}

The balance equation on the energy reads after integration by parts:
\begin{align}
    \frac{d {\mathcal H}}{dt}&= \int_a^b \frac{\delta {\mathcal H}}{\delta \bx} \frac{\partial {\mathcal \bx}}{\partial t} d\zeta \nonumber\\ &=  \left[ -\left(\frac{1}{2} v^2+h \right)\rho v +\mu v \partial_\zeta v + k \partial_\zeta T  \right]_{a}^b
\end{align}
encoding the first law of Thermodynamics, i.e. $\frac{d {\mathcal H}}{dt}=0$ if the boundary conditions are set to zero. 
The balance equation on the entropy reads:
\begin{align}
    \frac{d {\mathcal S}}{dt}&= \int_a^b \frac{\delta {\mathcal S}}{\delta \bx} \frac{\partial {\mathcal \bx}}{\partial t} d\zeta \nonumber\\ & = \int_a^b \left( -s \partial_\zeta (\rho v) -  \rho v \partial_\zeta s +\frac{\mu}{T} (\partial_\zeta v)^2+\frac{k}{T} \partial_\zeta \left( \partial_\zeta v\right) \right) d\zeta \nonumber \\
    &=  \int_a^b \left( \frac{\mu}{T}  \left( \partial_\zeta v\right)^2 + \frac{k}{T^2} \left( \partial_\zeta T\right)^2 \right)d\zeta + \left[  \frac{k}{T} \partial_\zeta T -\rho v s \right]_a^b
\end{align}
encoding the second law of Thermodynamics, i.e. $\frac{d {\mathcal S}}{dt}\geq 0$ if the boundary conditions are set to zero. 

We will see in Section \ref{Sec:IPHSgeneral} that the system \eqref{IPHSEulerianfin} with appropriate definition of the input/output map fits with the generalisation of the IPHS formulations proposed in \citep{Ramirez2022CES} to systems with convective flows.

\subsection{IPHS formulation in Lagrangian coordinates}

Consider now the same compressible, non-isentropic fluid in Lagrangian coordinates whose dynamics is goverened by the equations \eqref{eq:density_Lagrange}, \eqref{eq:velocity_Lagrange}, and \eqref{eq:entropy_Lagrange}. We choose as state variables $\bz := [z^\top\; s]^\top$, where $z := [\phi := 1/\rho \; v]^\top$. 
The Hamiltonian is given by:
\[
\cH(\bz) = \int_{a}^{b} \left( \frac12 \rho v^2 + \rho u \right) d\zeta
      = \int_{a}^{b} \left( \frac12 v^2 + u \right) \rho_0\, d\eta.
\]

Then the co-energy variables are:
\[
\delta_\phi \cH = -\rho_0 p, \qquad
\delta_v \cH = \rho_0 v, \qquad\delta_s \cH = \rho_0 T.
\]
It can be easily verified that in Lagrangian coordinates, the viscous stress reads $\tau = -\frac{\mu}{\phi \rho_0}\partial_\eta v$, and the heat flux $q = -\frac{k}{\phi \rho_0}\partial_\eta T$.

We can now rewrite the governing equations in terms of energy and co-energy variables in Lagrangian coordinates as

\begin{align*}
    \partial_t \phi& = \frac{1}{\rho_0}\partial_\eta\left(\frac{1}{\rho_0} \partial_v \cH\right),\\
\partial_t v &= \frac{1}{\rho_0}\partial_\eta\left(\frac{1}{\rho_0}\delta_\phi \cH\right)
         +\frac{1}{\rho_0} \partial_\eta\left( \frac{\mu}{\phi T\rho_0^2}\partial_\eta v \delta_s \cH\right),
\\
\partial_t s &= \frac{\mu}{T \phi \rho_0^2} \partial_\eta v \partial_\eta \left( \frac{\delta_v \cH}{\rho_0} \right) + \frac{1}{\rho_0 T}\partial_\eta\left( \frac{k}{\phi\rho_0}\, \partial_\eta \left(\frac{\delta_s \cH }{\rho_0} \right)\right).
\end{align*}

Which is equivalent to:
\begin{equation}
\label{eqdiffoperatorE}
\begin{pmatrix}
\partial_t \phi \\[2mm]
\partial_t v \\[2mm]
\partial_t s
\end{pmatrix}
=
\begin{pmatrix}
0 & \dfrac{1}{\rho_0}\partial_\eta\left(\dfrac{1}{\rho_0}\cdot \right) & 0 \\
\dfrac{1}{\rho_0}\partial_\eta\left(\dfrac{1}{\rho_0}\cdot\right) & 0 & 
    \cA_1 \\[3mm]
0 &  \cB_1& \cC_1
\end{pmatrix}
\begin{pmatrix}
\delta_\phi \cH \\[2mm]
\delta_v \cH \\[2mm]
\delta_s \cH
\end{pmatrix}
\end{equation}

Where the operators $\cA_1$ and $\cB_1$ are defined from Hilbert space $H^1({[a,b]})$ to $H^1({[a,b]})$ as follows:
\begin{align}
\cA_1 &: f \mapsto \frac{1}{\rho_0} \partial_\eta\left( \frac{\mu}{\Phi T\rho_0^2}\partial_\eta v f \right)\\
\cB_1 &: g  \mapsto \frac{\mu \partial_\eta v}{T \phi \rho_0^2}  \partial_\eta \left( \frac{1}{\rho_0} g \right),
\end{align}
and the operator $\cC_1$ is given by 
\begin{align}
    \cC_1\delta_sH=\frac{1}{\rho_0 T}\partial_\eta\left( \frac{k}{\phi\rho_0}\, \partial_\eta \left(\frac{1 }{\rho_0} \delta_sH \right)\right).
\end{align}
\begin{lemma}
    The operator $\cB_1$ is the formal adjoint of $\cA_1$, i.e. $\cB_1=\cA_1^\ast$ . Moreover, $\cC_1$ is formally skew-symmetric. 
\end{lemma}
\begin{proof}
    It is easy to show that $\frac{\mu \partial_\eta v}{T \Phi \rho_0^2}  \partial_\eta \left( \frac{1}{\rho_0} \cdot \right)$ is the formal adjoint of $ \frac{1}{\rho_0} \partial_\eta\left( \frac{\mu}{\Phi T\rho_0^2}\partial_\eta v \cdot \right)$ by using integration by parts considering omogeneous boundary conditions. The skew adjointness of $\frac{1}{\rho_0 T}\partial_\eta\left( \frac{k}{\Phi\rho_0}\, \partial_\eta \left(\frac{1 }{\rho_0} \cdot\right)\right)$ can be proven in a similar way.
\end{proof}


Equation \eqref{eqdiffoperatorE} is on the form 
\begin{equation}
\label{eqdiffoperatorfin}
\begin{pmatrix}
\partial_t \phi \\
\partial_t v \\
\partial_t s
\end{pmatrix}
=
\begin{pmatrix}
0 & \dfrac{1}{\rho_0}\partial_\eta\left(\dfrac{1}{\rho_0}\cdot \right) & 0 \\
\dfrac{1}{\rho_0}\partial_\eta\left(\dfrac{1}{\rho_0}\cdot\right) & 0 & 
    \cA_1 \\[3mm]
0 &  \cA_1^\ast& \cC_1
\end{pmatrix}
\begin{pmatrix}
\delta_\phi \cH \\[2mm]
\delta_v \cH \\[2mm]
\delta_s \cH
\end{pmatrix}
\end{equation}

The balance equation on the energy reads after integration by parts:
\begin{equation}
    \frac{d {\mathcal H}}{dt}= \int_a^b \frac{\delta {\mathcal H}}{\delta \bx} \frac{\partial {\mathcal \bx}}{\partial t} d\zeta=  \left[ -pv+\mu v \partial_\zeta v + f \partial_\zeta T  \right]_{a}^b
\end{equation}

encoding the first law of Thermodynamics, i.e. $\frac{d {\mathcal H}}{dt}=0$ if the boundary conditions are set to zero. 
The balance equation on the entropy reads:
\begin{align}
    \frac{d {\mathcal S}}{dt}&= \int_a^b \frac{\delta {\mathcal S}}{\delta \bx} \frac{\partial {\mathcal \bx}}{\partial t} dz \nonumber\\ & =\frac{\mu}{\Phi \rho_0 T}\left(\partial_\eta v\right)^2+ \frac{k}{\Phi \rho_0 T^2}\left(\partial_\eta T\right)^2+\left[ \frac{k}{T}\partial_\eta T\right]_a^b
\end{align}
encoding the second law of Thermodynamics, i.e. $\frac{d {\mathcal S}}{dt}\geq 0$ if the boundary conditions are set to zero. 

In this case the system \eqref{eqdiffoperatorfin} with appropriate definition of the input/output boundary maps fits with the formulations given in \citep{Ramirez2022CES} and their generalisation to convective systems proposed in Section \ref{Sec:IPHSgeneral}.

\section{Boundary controlled IPHS}
\label{Sec:IPHSgeneral}

Building on the previous examples, the IPHS definition for boundary-controlled systems in \citep{Ramirez2022CES} can be extended to include irreversible phenomena and convective flows. 
\begin{Definition}
\label{definition_IPHS}
    A general structure of an infinite-dimensional irreversible port-Hamiltonian system with potential convective flows undergoing $m$ irreversible processes is defined by
    \begin{align}\label{eq:IPHS_general_structure}
        \begin{bmatrix}
            \partial_t x\\ \partial_t s
        \end{bmatrix}&=\begin{bmatrix}
            P_0 & G_1 C(\bx)+G_0\boldsymbol{R}_0\\ 
            -G_1^\top C(\bx)-\boldsymbol{R}_0^\top G_0^\top & 0
        \end{bmatrix}
        \begin{bmatrix}
            \delta_x \cH\\ \delta_s \cH
        \end{bmatrix} \nonumber\\
        &+ \begin{bmatrix}
            P_1\partial^{f_0}_\zeta(\cdot) &G_1\partial^{f}_\zeta(\bR_1\cdot)\\-\bR_1^\top \partial^{f}_\zeta(G_1^\top \cdot) & r_s\partial^{f}_\zeta (g_s \cdot) + g_s\partial^{f}_\zeta (r_s\cdot)
        \end{bmatrix}\begin{bmatrix}
            \delta_x \cH\\ \delta_s \cH
        \end{bmatrix},
    \end{align}
where $\partial^{f}_\zeta(g)=f(x)\partial_\zeta(f(x)g)$, $\bR_0, \bR_1 \in \rline^{m\times 1}$, $r_s \in \rline$ are the modulating functions defined as follows
    \begin{align*}
        R_{0,i} &= \gamma_{0,i}(\bx,\delta_x \cH,\delta_s \cH)\{S|G_0(:,i)|\cH\},\\
        R_{1,i} &= \gamma_{1,i}(\bx,\delta_x \cH,\delta_s \cH)\{S|G_1(:,i)\partial^{f}_\zeta|\cH\},\\
        r_s &= \gamma_s(\bx,\delta_x \cH,\delta_s \cH)\{S|\cH\},
    \end{align*}
$\gamma_{0,i}, \gamma_{1,i}, \gamma_s$ are strictly positive-valued functions. $G_0, G_1 \in \rline^{n-1 \times m}, C(\bx) \in \rline^{m \times 1}$. The pseudo-Poisson brackets are defined as
\begin{align*}
    \{\cE|\cG|\cF\} &= \begin{bmatrix}
        \delta_{x}\cE\\\delta_s\cE
    \end{bmatrix}^\top \begin{bmatrix}
        0 & \cG\\ -\cG^\ast & 0
    \end{bmatrix}\begin{bmatrix}
        \delta_{x}\cF\\\delta_s\cF
    \end{bmatrix},\\
    \{\cE|g_s(x)|\cF\} &= \delta_s \cE \partial_\zeta(g_s(x)\delta_s\cF)
\end{align*}
for some smooth function $\cE$ and $\cF$, operator $\cG$ and its formal adjoint $\cG^\ast$ and a real continuous scalar function $g_s(x)$.
\end{Definition}

From this definition one can parametrize sets of boundary input and output in order to define the considered class of boundary controlled IPHS. 

\begin{Assumption}\label{Assumption_1}
The operator 
\[
\begin{bmatrix}
\delta_x \cH\\ 
\delta_s \cH
\end{bmatrix}^\top
\begin{bmatrix}
 P_1\partial^{f_0}_\zeta(\cdot) \\
-G_1^\top C(\bx)   
\end{bmatrix}
\cdot \delta_x \cH\
\]
is an exact differential form.
\end{Assumption}

\begin{Definition}\label{definition_BC_IPHS}
A Boundary Controlled IPHS (BC-IPHS) is an infinite dimensional IPHS according to Definition \ref{definition_IPHS} augmented with the boundary port variables
\begin{align}\label{eq:IPHS_input}
v(t)  & = W_{B}
\begin{bmatrix}
e(t,b) \\
e(t,a)
\end{bmatrix},&  &y(t)  = W_{C}
\begin{bmatrix}
e(t,b) \\
e(t,a)
\end{bmatrix} 
\end{align}
as linear functions of the modified effort variable
\begin{equation}\label{boundary_variables_IPHS}
e(t,z)  = 
\begin{bmatrix}
f_0(x)\frac{\delta H}{\delta x}(t,\zeta) \\
f(x) \mathbf{R}(\mathbf{x},\frac{\delta H}{\delta \mathbf{x}}) \:\frac{\delta H}{\delta s}(t,\zeta)  \\
\end{bmatrix},
\end{equation}
with $\mathbf{R}(\mathbf{x},\frac{\delta H}{\delta \mathbf{x}}) =\begin{bmatrix}1 & \mathbf{R_1(\mathbf{x},\frac{\delta H}{\delta \mathbf{x}})} & \mathbf{r_s(\mathbf{x},\frac{\delta H}{\delta \mathbf{x}})} \end{bmatrix}^\top$ and
\begin{align*}
W_{B} &= 
\begin{bmatrix}
\frac{1}{\sqrt{2}}\left(\Xi_2 + \Xi_1P_{ep} \right)M_p & \frac{1}{\sqrt{2}} \left(\Xi_2 - \Xi_1 P_{ep} \right)M_p
\end{bmatrix},
\\
W_{C}& = 
\begin{bmatrix} 
\frac{1}{\sqrt{2}}\left(\Xi_1+\Xi_2 P_{ep}\right)M_p  & \frac{1}{\sqrt{2}}\left(\Xi_1-\Xi_2 P_{ep}\right)M_p 
\end{bmatrix},
\end{align*}   
where $M_p=\left( M^\top M\right)^{-1}M^\top$, $P_{ep}=M^\top P_e M$ and $M\in \mathbb{R}^{(n+m+2) \times k} $ is spanning the columns of $P_e \in \mathbb{R}^{n+m+2}$ of rank $k$, defined by\footnote{$0$ has to be understood as the zero matrix of appropriate dimensions.}
\begin{equation}\label{def:Pe}
P_e=\begin{bmatrix}
 P_1 & 0 &G_1 & 0\\
0 & 0 &0 & g_s\\
G_1^\top & 0 &0 & 0\\
0 & g_s  & 0  & 0\\
\end{bmatrix}
\end{equation}
and where $\Xi_1$ and $\Xi_2$ in $\mathbb{R}^{k \times k} $satisfy $\Xi_2^\top\Xi_1+\Xi_1^\top\Xi_2=0$ and $\Xi_2^\top\Xi_2+\Xi_1^\top \Xi_1=I$.\rule{0.5em}{0.5em}
\end{Definition}

A direct consequence of the structure given in Definition \ref{definition_IPHS} and  \ref{definition_BC_IPHS} is the completion with the first and second law of Thermodynamics as stated in the following lemmas. 
\begin{lemma}[First law of Thermodynamics]
    The total energy balance is
    \begin{equation}
\frac{d \cH(t)}{dt}=y(t)^Tv(t)
\end{equation}
which leads, when the input is set to zero, to $\partial \cH =0$ in accordance with
the first law of Thermodynamics.
\end{lemma}

\begin{lemma}[Second law of Thermodynamics]
    The total entropy balance is given by
\begin{equation}
\frac{d \cS(t)}{dt}=\int_a^b \sigma_t d\xi + y_s(t)^T v_s(t)
\end{equation}    
where $y_s$ and $v_s$ are the entropy conjugated input/output and $\sigma_t$ is the
total internal entropy production. This leads, when the input is set to
zero, to $\frac{dS}{dt}=\int_a^b \sigma_t d\xi \geq 0$ in accordance with the second law of Thermodynamics. 
\end{lemma}
\begin{proof}
The proof follows the same development as in \citep{Ramirez2022CES} with the additional use of Assumption \ref{Assumption_1} which accounts for the coupling between the material and thermal domain due to the convection.
\end{proof}
The formulation of non-isentropic, non-reactive fluid defined in Eulerian and Lagrangian coordinates is given in the following propositions.

\begin{proposition}
A compressible, non-isentropic fluid in Eulerian coordinates can be described as an irreversible port-Hamiltonian system \eqref{eq:IPHS_general_structure} considering the following boundary inputs and outputs 
   \begin{align}\label{eq:boundary_inputs_Euler}
    \bu = \begin{bmatrix}
        - \left(\frac{1}{2}v^2+h\right)(b)-\frac{1}{\rho}\tau(b)\\
        -f_s(b)\\
         \left(\frac{1}{2}v^2+h\right)(a)+\frac{1}{\rho}\tau(a)\\
        f_s(a)
    \end{bmatrix}, \; \by = \begin{bmatrix}
         \rho v(b)\\
        T(b)\\
        \rho  v(a)\\
        T(a)
    \end{bmatrix}
\end{align}
respectively, then the total energy $\cH$ satisfies 
\begin{align*}
    \dot{\cH} = \by^\top \bu,
\end{align*}
where $f_s := \frac{k}{T}\partial_\zeta T$ is the entropy flux,
and the irreversible entropy production rate is
\begin{align*}
    \dot{\mathcal{S}}=\int_0^L \sigma_s d\zeta+\by_s^{\top} \bu_s,
\end{align*}
where $\cS$ is the total entropy of the fluid, $\bu_s$ and $\by_s$ are the boundary entropy input and output respectively. Moreover, $\sigma_s \geq 0$ represents the total internal entropy production.
\end{proposition}

\begin{proof}
For a non-isentropic, non-reactive fluid defined in Eulerian coordinates, we have $x=[\rho\;v]^\top$, $\cS = \int_a^b \rho s d\zeta, \; \delta_{\bx}\cS = [s\;0\;\rho]^\top.$ The matrices in \eqref{eq:IPHS_general_structure} are $P_0 = 0_{2\times 2}$, $P_1 = \begin{bmatrix}
    0 & -1\\-1 & 0
\end{bmatrix}$, $G_0 = 0_{2\times 1}$, $G_1 = \begin{bmatrix}
    0 & 1
\end{bmatrix}^\top$, $C(\bx) = \partial_\zeta s$, $g_s = 1$. We then have:
\[
P_e=\begin{bmatrix}
    0&-1&0&0&0 \\-1 &0 &0 &1&0\\0&0&0&0&1\\0&1&0&0&0 \\0&0&1&0&0
\end{bmatrix}, M=\begin{bmatrix}
    -\frac{1}{2}&0&0&0 \\0 &1 &0 &0\\0&0&0&1\\\frac{1}{2}&0&0&0 \\0&0&1&0
\end{bmatrix}
\]
\[
M_p=\begin{bmatrix}
    -1&0&0&1&0 \\0 &1 &0 &0&0\\0&0&0&0&1\\0&0&1&0&0 
\end{bmatrix}, P_{ep}=\begin{bmatrix}
    0&1&0&0 \\1&0&0&0\\0&0&0&1\\0&0&1&0
\end{bmatrix}
\]
We can compute $\bR_1 = \frac{\mu \partial_\zeta v}{T}$, $R_0 = 0$, $r_s = \frac{k \partial_\zeta T}{T^2}$, $\gamma_1 = \frac{\mu}{T} > 0$, $\gamma_0 = 0$, $\gamma_s = \frac{k}{T^2} > 0$. We have $f_0=1$ and $f=\frac{1}{\rho}$ leading to the extended set of effort variables
\[
e(t)= \begin{bmatrix}
\frac{1}{2}v^2+h\\\rho v \\ T \\ \frac{\mu}{\rho} \partial_\zeta v \\  \frac{k}{\rho T}\partial_\zeta T
\end{bmatrix}, M_p e(t)= \begin{bmatrix}
-\frac{1}{2}v^2-h+\frac{\mu}{\rho} \partial_\zeta v \\ \rho v \\ \frac{k}{\rho T}\partial_\zeta T\\ T
\end{bmatrix}
\]
leading by choosing $\Xi_1$ and $\Xi_2$ satisfying the condition of Definition \ref{definition_BC_IPHS} to the energy balance
\begin{align*}
    \dot{\cH} =& \int_a^b \left(\delta_\rho\cH\partial_t\rho + \delta_v\cH\partial_tv + \delta_s\cH\partial_ts \right) d\zeta\\
    =&\int_a^b \left( -\delta_\rho\cH\partial_\zeta\delta_v \cH\right. \\
    & -\delta_v\cH\partial_\zeta\delta_\rho\cH + \frac{1}{\rho}\delta_v\cH\partial_\zeta s \delta_s\cH - \frac{\delta_v\cH}{\rho}\partial_\zeta \tau\\
    & -\frac{1}{\rho}\delta_s\cH\partial_\zeta s \delta_v\cH -\tau\partial_\zeta \left(\frac{\delta_v\cH}{\rho}\right)\\ 
    &\left.- T\partial_\zeta \left( -\frac{q}{T}\right) - \frac{q}{T}\partial_\zeta T\right) d\zeta\\
    &= -\left.\left(\left(\frac{1}{2}v^2+h\right)\rho v\right)\right|_a^b - (\tau v)|_a^b - (Tf_s)|_a^b\\
    &=\by^\top \bu,
\end{align*}
where the boundary inputs and outputs are as depicted in \eqref{eq:boundary_inputs_Euler}. Next, we define the total entropy of the system as 
\begin{align*}
    \cS = \int_a^b \rho s d\zeta, \; \delta_{\bx}\cS = [s\;0\;\rho]^\top.
\end{align*}
Then the entropy balance goes as follows
\begin{align*}
    \dot{\cS} =& \int_a^b \delta_{\bx}\cS^\top \partial_t\bx \;d\zeta\\
    =& \int_a^b \left(-s\partial_\zeta (\rho v) - \rho v\partial_\zeta s - \frac{\tau}{T}\partial_\zeta v - \frac{1}{T}\partial_\zeta T\right)d\zeta\\
    =& \int_a^b \left[ -s\partial_\zeta (\rho v) - \rho v\partial_\zeta s - \partial_\zeta\left(\frac{q}{T}\right)\right] d\zeta\\ 
    &+ \int_a^b \left[\frac{\mu}{T}\left(\partial_\zeta\frac{\delta_v\cH}{\rho}\right)^2
    + \frac{k}{T^2}\left(\partial_\zeta \frac{\delta_s\cH}{\rho}\right)^2\right]d\zeta\\
    =& \int_a^b \sigma_s \; d\zeta - \left.\left(s \rho v + \frac{q}{T}\right)\right|_a^b,
\end{align*}
where $\sigma_s := \frac{\mu}{T}\left(\partial_\zeta\frac{\delta_v\cH}{\rho}\right)^2 + \frac{k}{T^2}\left(\partial_\zeta \frac{\delta_s\cH}{\rho}\right)^2 \geq 0 $ is the internal entropy production, $\by_s^\top \bu_s := \left.\left(s\rho v + \frac{q}{T}\right)\right|_a^b$.
\end{proof}

\begin{proposition}
    A compressible, non-isentropic fluid in Lagrangian coordinates can be described as an irreversible port-Hamiltonian system \eqref{eq:IPHS_general_structure} considering the following boundary inputs and outputs 
    
    \begin{align}\label{eq:boundary_inputs_Lagrangian}
    \bu = \begin{bmatrix}
        -p(b) - \tau(b)\\
        -f_s(b)\\
        p(a) + \tau(a)\\
        f_s(a)
    \end{bmatrix}, \; \by = \begin{bmatrix}
        v(b)\\
        T(b)\\
        v(a)\\
        T(a)
    \end{bmatrix}
\end{align}
respectively, then the total energy $\cH$ satisfies 
\begin{align*}
    \dot{\cH} = \by^\top \bu,
\end{align*}
and the irreversible entropy production rate is
\begin{align*}
    \dot{\mathcal{S}}=\int_0^L \sigma_s d\eta+\by_s^{\top} \bu_s,
\end{align*}
where $\bu_s$ and $\by_s$ boundary entropy input and output respectively. Moreover, $\sigma_s \geq 0$ represents the total internal entropy production.
\end{proposition}

\begin{proof}
For a non-isentropic, non-reactive fluid defined in Lagrangian coordinates, we have $x= [\phi\;v]^\top$, $\cS = \int_a^b \rho_0 s d\eta, \; \delta_{\bx}\cS = [0\;0\;\rho_0]^\top$. The underlying matrrices are $P_0 = 0_{2\times 2}$, $P_1 = \begin{bmatrix}
    0 & 1\\ 1 & 0
\end{bmatrix}$, $G_0 = 0_{2\times 1}$, $G_1 = \begin{bmatrix}
    0 & 1
\end{bmatrix}^\top$, $C(\bx) = 0$, $f_0=f=\frac{1}{\rho_0}$, $\bR_1 = \frac{\mu \partial_\zeta v}{\rho_0 T}$, $R_0 = 0$, $g_s = 1/\rho_0$, $r_s = \frac{k \partial_\zeta T}{\rho T^2}$, $\gamma_1 = \mu/\rho_0^2 T > 0$, $\gamma_0 = 0$, $\gamma_s = k/\rho_0^2 T^2 > 0$. Considering the state variables as $\bx := [\phi\;v\;s]^\top$, the energy balance in Lagrangian coordinates can be expressed as follows
   \begin{align*}
       \dot{\cH} = &\int_a^b (\delta_\phi\cH\partial_t\rho + \delta_v\cH\partial_t v + \delta_s\cH\partial_t s) d\eta\\
       =& -(pv + \tau v)|_a^b - (T f_s)|_a^b\\
       =& \by^\top \bu,
   \end{align*}
where the boundary inputs and outputs are depicted in \eqref{eq:boundary_inputs_Lagrangian}. Consider the total entropy of the system as 
\begin{align*}
    \cS = \int_a^b \rho_0 s d\eta, \; \delta_{\bx}\cS = [0\;0\;\rho_0]^\top.
\end{align*}
Then the entropy balance goes as follows
\begin{align*}
    \dot{\cS} =& \int_a^b \delta_{\bx}\cS^\top \partial_t\bx \;d\eta\\
    =& \int_a^b \left(- \frac{\tau}{T}\partial_\eta v - \frac{1}{T}\partial_\eta T\right)d\eta\\
    =& \int_a^b \left[-\partial_\eta\left(\frac{q}{T}\right)\right] d\eta\\ 
    &+ \int_a^b \left[\frac{\mu}{\phi \rho_0 T}\left(\partial_\eta\frac{\delta_v\cH}{\rho}\right)^2
    + \frac{k}{\phi \rho_0 T^2}\left(\partial_\eta \frac{\delta_s\cH}{\rho}\right)^2\right]d\eta\\
    =& \int_a^b \sigma_s \; d\eta - \left.\left(\frac{q}{T}\right)\right|_a^b,
\end{align*}
where $\sigma_s := \frac{\mu}{\phi \rho_0 T}\left(\partial_\eta\frac{\delta_v\cH}{\rho}\right)^2 + \frac{k}{\phi \rho_0T^2}\left(\partial_\eta \frac{\delta_s\cH}{\rho}\right)^2 \geq 0 $ is the internal entropy production, $\by_s^\top \bu_s := \left.\left(\frac{q}{T}\right)\right|_a^b$.
\end{proof}

\section{Conclusion}
\label{Sec:conclusions}

After recalling the constitutive relations of non-isentropic fluids in both Eulerian and Lagrangian coordinates, we show how these systems can be cast as Irreversible Port-Hamiltonian Systems by extending the definition in \cite{Ramirez2022CES} to include convective transport. The key modification concerns the differential operators, whose structure introduces convective terms that must be treated appropriately. As in \cite{Ramirez2022CES}, we then derive a parametrisation of the boundary port variables that ensures consistency with the first and second laws of thermodynamics.

\bibliography{sample,references_HM}

\end{document}